\documentclass[12pt]{amsart}
\numberwithin{equation}{section}

\usepackage{amssymb}
\def\ca{{\mathcal A}}
\def\cb{{\mathcal B}}

\def\cd{{\mathcal D}}
\def\ce{{\mathcal E}}
\def\cf{{\mathcal F}}

\def\ch{{\mathcal H}}

\def\cn{{\mathcal N}}

\def\cs{{\mathcal S}}

\def\cx{{\mathcal X}}

\def\ga{{\mathfrak A}} \def\gpa{{\mathfrak a}}
\def\gb{{\mathfrak B}}

\def\gps{{\mathfrak s}}
 \def\gpt{{\mathfrak t}}

\def\gz{{\mathfrak Z}}

\def\bc{{\mathbb C}}

\def\bm{{\mathbb M}}
\def\bn{{\mathbb N}}

\def\br{{\mathbb R}}

\def\bz{{\mathbb Z}}

\def\a{\alpha}
\def\b{\beta}
\def\g{\gamma}  \def\G{\Gamma}
\def\d{\delta}  \def\D{\Delta}
\def\eeps{\epsilon}
\def\eps{\varepsilon}
\def\l{\lambda} \def\L{\Lambda}

\def\m{\mu}

\def\n{\nu}

\def\s{\sigma} \def\S{\Sigma}
\def\t{\tau}
\def\f{\varphi} \def\ff{\phi} \def\F{\Phi}
  \def\Th{\Theta}
\def\om{\omega} \def\Om{\Omega}

\def\id{\hbox{id}}

\newtheorem{thm}{Theorem}[section]
\newtheorem{lem}[thm]{Lemma}

\newtheorem{cor}[thm]{Corollary}
\newtheorem{prop}[thm]{Proposition}
\newtheorem{defin}[thm]{Definition}

\def\aut{\mathop{\rm Aut}}
\def\carf{\mathop{\rm CAR}}
\def\e{\mathop{\rm Exp}\,}

\def\esssup{\mathop{\rm esssup}}

\def\sp{\mathop{\rm sp}}

\def\supp{\mathop{\rm supp}}
\def\dist{\mathop{\rm dist}}
\newcommand{\ty}[1]{\mathop{\rm {#1}}}
\def\di{\mathop{\rm d}\!}
\def\vol{\mathop{\rm vol}}
\def\ad{\mathop{\rm ad}}
\def\idd{{1}\!\!{\rm I}}

\newcommand{\nn}{\nonumber}

\begin{document}

\title[disordered fermions]
{disordered fermions on lattices and their spectral properties}
\author{Stephen Dias Barreto}
\address{Stephen Dias Barreto\\
Department of Mathematics \\
Padre Conceicao College of Engineering\\
Verna Goa 403 722, India} \email{{\tt sbarreto@pccegoa.org}}
\author{Francesco Fidaleo}
\address{Francesco Fidaleo\\
Dipartimento di Matematica \\
Universit\`{a} di Roma Tor Vergata\\
Via della Ricerca Scientifica 1, Roma 00133, Italy} \email{{\tt
fidaleo@mat.uniroma2.it}}
\date{\today}

\begin{abstract}
We study Fermionic systems on a lattice with random interactions
through their dynamics and the associated KMS states. These systems
require a more complex approach  compared with the standard spin
systems on a lattice, on account of  the difference in commutation
rules for the local algebras for disjoint regions, between these two
systems. It is for this reason that some of the known formulations
and proofs in the case of the spin lattice systems with random
interactions do not automatically go over to the case of disordered
Fermion lattice systems.
%As in the previous case whose the observable algebra is made of the tensor product of infinitely many copies of full matrix algebras,
We  extend to the disordered CAR algebra, some standard results
concerning the spectral properties exhibited by temperature states
for  disordered quantum spin systems. We discuss the Arveson
spectrum and its connection with the Connes and Borchers
$\G$--invariants for such $W^{*}$--dynamical systems. In the case of
KMS states exhibiting a natural property of invariance with respect
to the spatial translations, some interesting properties, associated
with standard spin--glass--like  behaviour, emerge naturally. It
covers infinite--volume limits of finite--volume Gibbs states, that
is the quenched disorder for Fermions living on a standard lattice
$\bz^d$. In particular, we show that a temperature state of the
systems under consideration can generate only a type $\ty{III}$ von
Neumann algebra (with the type $\ty{III_{0}}$ component excluded).
Moreover, in the case of the pure thermodynamic phase, the
associated von Neumann is of type $\ty{III_{\l}}$ for some
$\l\in(0,1]$, independent of the disorder. Such a result is in
accordance with the principle of self--averaging which affirms that
the physically relevant quantities do not depend on the disorder.
The present approach can be viewed as a further step towards  fully
understanding the very complicated structure of the set of
temperature states of quantum spin glasses, and its connection with
the breakdown of the symmetry for the replicas.
\end{abstract}

\maketitle

\section{introduction}
\label{sec1}

Interacting Fermion systems on a lattice have usually been studied by considering a spinless
Fermions at each lattice site which interact with each other.
The restriction to spinless particles is just a  matter of
simplification of notation and more general situations can be treated as well. Investigations concerning the existence
of dynamics have been made in the past  and, more recently, the
equilibrium statistical mechanics of such systems including the
thermodynamic limits have  been studied. We refer the reader to
\cite{AM1, AM2, M, M1} and the literature cited therein, for a systematic treatment of the topic.

An example of a Fermion lattice system is the Hubbard model (see
e.g. \cite{Tk, WS}) which describes electrons in a solid,
interacting with each other through a repulsive Coulomb force. The
spinless counterpart was considered in \cite{M2}, and some
particular case of its disordered version is analyzed in some detail
in Section \ref{ccoomoo} of the present paper. Other (non disordered) models based on
Fermions, and connected with the Quantum Markov Property and the
Quantum Information Theory, are considered in \cite{AFM, F1}.
Importantly, in the majority of the cases the interaction potential
is assumed to be even. Of course, there are situations wherein the
potential considered is non even (cf. \cite{A2}), but it is yet
unclear if non even interactions have relevant physical
applications.

Another very important line of research in Statistical Mechanics is
that involving the so called {\it spin glasses}, falling into the
more general category of disordered systems. The first model,
constucted on tensor product of copies of a single algebra of
observables localized on the sites in a lattice, is the
Sherrington--Kirkpatrick model (cf. \cite{SK}). It is a disordered
mean field model for which it is meaningless to define the dynamics in the thermodynamic limit.
A more realistic model is the so called Edwards--Anderson model (cf.
\cite{EA}), which can be considered as a disordered generalization of the Ising or ferrimagnetic
model, provided the distribution of the coupling constants is one--sides.\footnote{The most interesting situation, corresponding to a spin glass, is when the distribution of the coupling constants is one--sides.}
The investigation of the {\it quenched disorder} for spin
glasses is a fairly formidable task. Among the problems which are
still open, we mention the {\it the breakdown of the symmetry for
replicas}. For the convenience of the reader, we report the
references \cite{B1, EH1, EH, Gu, HP1, HP2, MPV, Ne, NeS1, Ta, Y}
which are just a sample (far from being  complete), of some of the
work done  on the theory of the spin glasses.\footnote{A very nice
explanation of the failure of the Replica Method Solution, in terms of the
Moment Problem is given in \cite{Tn}.}

It is then natural to undertake the study of the disordered
systems for models which include  Fermions. For models without
Fermions, the study of such disordered systems was firstly carried out in
\cite{K}, using the standard techniques  of Operator Algebras. Apart
from the general properties of such disordered systems established
in that paper, the problem of the so called {\it weak Gibbsianess},
that is the appearance of weaky Gibbsian states which are not
jointly Gibbsian with respect to the observable variables and the
coupling constants taken together, is well explained. The reader is
referred to \cite{EMSS, Ku} for some concrete example on weak
Gibbsianess relative to the classical case. In \cite{B, BF}, general
properties of temperature (i.e. KMS) states, and their spectral
properties were studied in detail. Finally, in \cite{F}, the theory
of chemical potential is extended to such disordered systems. The
reader is referred also to \cite{A1, BF1} for a good review on the
topic.

In our model we consider Fermions on a lattice with random even
interactions between the spinless particles located at the lattice
sites. It is expected that the spectra of the random evolution group
of this infinite Fermion system will exhibit some invariance
properties. Besides, the invariant KMS states are also expected to
enjoy some nice structural properties. Because of the complex random
structure and the (anti)commutation properties of local algebras,
the analysis of such a  system is a fairly daunting task. Some of
the known formulations and proofs for disorderd spin lattice systems
do not automatically go over to the disorderd Fermion lattice
systems. 

By using a standard procedure (cf. e.g. \cite{BES}), we
start with an appropriate $C^{*}$--algebra of observables, that is
$\ga:=\ca\otimes L^{\infty}(\Om,\m)$. In order to encode the
Fermions, we consider  a separable unital
$C^*$--algebra $\ca$, equipped with a $\bz_2$--grading. In
particular, in some concrete examples $\ca$ will be the CAR algebra $\carf(\bz^d)$ on the
lattice $\bz^d$. In order to take into account the disorder, the
probability space $(\Om,\m)$ is the sample space for the coupling
constants, the latter being random variables on it. For such a
disordered system, the lattice translations and the time evolution
act in a natural way as mutually commuting group actions. The
resulting systems fall into the category of so called graded
asymptotically Abelian systems. Due to the grading, the study of the
spectral properties of such systems is  more involved than that of
the asymptotically Abelian ones. 

After investigating the general
properties of the disordered systems (sections \ref{sec2},
\ref{gmlc}, \ref{sssttt}), in Section \ref{asgab} we generalize some spectral properties, known for
asymptotically Abelian systems, to
the $\bz_2$--graded models under considerations. Section \ref{dsaspt} is devoted to apply such spectral results obtained in the graded situation, to the study of the structure of the von Neumann algebras generated by $\bz_2$--graded asymptotically Abelian dynamical systems.
Then we are able to investigate the type of the von Neumann algebras arising from temperature states of the Fermionic systems under consideration. The reader is referred to 
\cite{A, BF, D, HL, L, Se} for the analogous results known for asymptotically Abelian dynamical systems. 

The first general result we are able to prove is that a von Neumann
algebra with a nontrival infinite semifinite summand cannot carry an
action which is graded asymptotically Abelian w.r.t. the strong
topology. Namely, we generalize the corresponding result known for
asymptotically Abelian systems. We then pass on to the investigation
of the Arveson spectrum and its connection with the Connes and
Borchers $\G$--invariants for the (non factor) $W^{*}$--dynamical
systems equipped with a $\bz_2$--grading. We apply such results to
$W^{*}$--dynamical systems generated by the GNS representation of
temperature states exhibiting natural equivariance properties with
respect to the spatial translations and the time evolution. Then
some interesting properties, associated with  the standard
spin--glass--like behaviour, emerge naturally. The analysis covers the
case of KMS states obtained by infinite volume limits of
finite--volume Gibbs states, that is the quenched disorder for
Fermions living on a standard lattice $\bz^d$. We mention the fact
that a temperature state of such disordered Fermions can generate
only a type $\ty{III}$ von Neumann algebra, with the type
$\ty{III_{0}}$ component excluded.

As explained in \cite{BF}, the natural candidate for  the pure
thermodynamic phase is when the center
$\pi_\f(\ga)''\bigwedge\pi_\f(\ga)'$ of the GNS representation of a
KMS state $\f$, is "as  trivial as possible", that is
$$
\pi_\f(\ga)''\bigwedge\pi_\f(\ga)'\sim L^{\infty}(\Om,\m)\,.
$$ 
Even
for disordered systems incuding Fermions, a consequence of the
previously metioned result is that for a pure thermodynamic phase $\f$, $\pi_\f$ generates a
$\ty{III_{\l}}$ von Neumann algebra, for some $\l\in(0,1]$,
independent of the disorder. Namely, for the pure thermodynamic
phase of the disordered models under consideration, $\pi_\f(\ga)\sim
M\overline{\otimes}L^{\infty}(\Om,\m)$, where $M$ is the unique type
$\ty{III_{\l}}$ hyperfinite von Neumann factor. Such a result is in
accordance with the principle of self--averaging which affirms that
the physically relevant quantities do not depend on the
disorder. For a nice explanation on the study of the
spectral properties and the connected investigation of the type of
the factors appearing in Quantum Statistical Mechanics, the reader
is referred to the review paper \cite{Ka} and the literature cited
therein. We cite also the paper \cite{L1} where an interesting connection between the modular structure of the algebras of the observables and the statistics of the black holes is established.

The paper ends with a section devoted to the detailed analysis of a concrete model
based on a kind of disordered spinless Fermions, which reduces itself to the disordered Hubbard Hamiltonian, provided the distribution of the coupling constants is one--sides.

The symmetry replica breaking is one of the most important open problems in the theory of the spin glasses. As our approach is naturally based on the replicas, one for each value of the coupling constants, we hope that the approach followed in the present paper, as well as
that in the previous connected works \cite{B, BF, F, K},
can be viewed as a significant step towards  fully understanding
the very complicated structure of the set of temperature states of
quantum spin glasses, and its connection with the breakdown of the
symmetry for  replicas.

\section{the description of the model}
\label{sec2}

In the present paper we deal only with von Neumann algebras with
separable preduals unless specified otherwise. Besides, all
representations of the involved $C^{*}$--algebras are understood to
act on separable Hilbert spaces. Finally, all our $C^*$--algebras
have the identity $\idd$. Denote by $[a,b]:=ab-ba$,
$\{a,b\}:=ab+ba$, the commutator and anticommutator between elements
$a$, $b$, respectively.

We start by quickly reviewing the basic properties of the Fermion $C^{*}$-algebra $\carf(\bz^d)$ on a lattice
$\bz^{d}$.
%It is generated by the annihilation and creation operators satisfying the Canonical
%Anticommutation Relations.
Indeed, let $J$ be any set. The {\it Canonical Anticommutation Relations} (CAR for short) algebra
over  $J$ is the $C^{*}$--algebra $\carf(J)$ with the identity $\idd$
generated by the set $\{a_j, a^{\dagger}_j\}_{j\in I}$ (i.e. the Fermi annihilators and creators respectively), and the relations
\begin{equation*}
(a_{j})^{*}=a^{\dagger}_{j}\,,\,\,\{a^{\dagger}_{j},a_{k}\}=\d_{jk}\idd\,,\,\,
\{a_{j},a_{k}\}=\{a^{\dagger}_{j},a^{\dagger}_{k}\}=0\,,\,\,j,k\in J\,.
\end{equation*}
On the CAR algebra the parity automorphism $\Th$ acts on the generators as
$$
\Th(a_{j})=-a_{j}\,,\,\,\Th(a^{\dagger}_{j})=-a^{\dagger}_{j}\,,\quad j\in J\,,
$$
and induces on $\carf(J)$ a $\bz_{2}$--grading. This grading yields,
$\carf(J)=\carf(J)_{+} \oplus\carf(J)_{-}$ where
\begin{align*}
&\carf(J)_{+}:=\{a\in\carf(J) \ | \ \Th(a)=a\}\,,\\
&\carf(J)_{-}:=\{a\in\carf(J) \ | \ \Th(a)=-a\}\,.
\end{align*}
Elements  in $\carf(J)_+$ and in $\carf(J)_-$ are called
{\it even} and {\it odd}, respectively.

A map $T:\ca_1\to\ca_2$ between $C^*$--algebras with $\bz_{2}$--gradings
$\Th_1$, $\Th_2$ is said to be
{\it even} if it
is grading--equivariant:
$$
T\circ\Th_1=\Th_2\circ T\,.
$$
The previous definition applied to states $\f\in\cs(\carf(J))$ leads to $\f\circ\Th=\f$, that is $\f$ is even if it is
$\Th$--invariant.

Let the index set $J$ be countable, then the CAR algebra  is isomorphic to the
$C^{*}$--infinite
tensor product of $J$--copies of $\bm_{2}(\bc)$:
\begin{equation}
\label{jkw}
\carf(J)\sim\overline{\bigotimes_{J}\bm_{2}(\bc)}^{C^*}\,.
\end{equation}
Such an isomorphism is established by the Jordan--Klein--Wigner
transformation, see e.g. \cite{T3}, Exercise XIV. When $J=\bz^d$,
the above mentioned isomorphism does not preserve the canonical
local properties of the CAR algebra, thus it cannot be used to
investigate the local properties of the model.

Thanks to \eqref{jkw}, $\carf(J)$ has a unique tracial state $\t$
%(at least when $J$ is countable)
as the extension of the unique tracial state on $\carf(I)$,
$|I|<+\infty$. Let $J_1\subset J$ be a finite set and
$\f\in\cs(\carf(J))$. Then there exists a unique positive element
$T\in\carf(J_1)$ such that
$\f\lceil_{\carf(J_1)}=\t\lceil_{\carf(J_1)}(T\,{\bf\cdot}\,)$. The
element $T$ is called the {\it adjusted matrix} of
$\f\lceil_{\carf(J_1)}$. For the standard applications to quantum
statistical mechanics, one  also uses the density matrix w.r.t. the
unnormalized trace.

Our aim is to investigate disordered models of Fermions on lattices.
Our starting point will be the Fermion algebra $\carf(\bz^d)$
together with a (formal) random Hamiltonian. We denote
$\carf(\L)\subset\carf(\bz^d)$ the local CAR subalgebra generated by
$\{a_j, a^{\dagger}_j\mid j\in\L\}$. We can then  consider a net
$\{H_{\L}(\om)\}_{\L\subset\bz^{d}}$, $\L$ being any finite subsets
of $\bz^{d}$, of local random Hamiltonians which are even with
respect to the parity automorphism $\Th$, which is made up of
$\carf(\bz^d)_{s.a.}$--valued measurable functions arising from
finite--range even interactions. The net under consideration
satisfies the equivariance condition
\begin{equation}
\label{6}
H_{\L+x}(\om)=\a_{x}(H_{\L}(T_{-x}\om))\,.
\end{equation}
Such a picture arises naturally in the study of disordered systems
(see e.g. \cite{BES, K}), and more precisely when one considers
Fermion systems with random, even Hamiltonians. A concrete model
arising from a random Hamiltonian as in \eqref{6} is described in some 
detail in Section \ref{ccoomoo}.

Associated with such a random Hamiltonian, there is a one parameter
group of random automorphisms $\t_{t}^{\om}$ of $\carf(\bz^d)$, one
for each choice of the coupling constants in the sample space
$\Om$.\footnote{The reader is referred to the seminal paper
\cite{AM1} concerning the statistical mechanics associated with (non
disordered) Fermions.} As described below, we assume that
$\t_{t}^{\om}$ enjoys good joint measurability conditions, and local
properties, see Section \ref{gmlc}.
%It is also natural to assume that $\t_{t}^{\om}$ acts locally on constant elements of the form
%$A\otimes\idd$, $A\in\carf(\bz^d)$. Namely, $\t_{t}^{\om}(A)\in\carf(\bz^d)\otimes L^{\infty}(\Om,
%\m)$, see again Section \ref{ccoomoo}.
%the function $f_{A,t}$ given by
%\begin{equation*}
%f_{A,t}(\om):=\t_{t}^{\om}(A)\,.
%\end{equation*}
In view of the possible applications to more general situations
including disordered gauge theories on the lattices and/or
disordered theories arising from quantum field theory, our main
object will be merely a unital $\bz_2$--graded separable
$C^*$--algebra $\ca$. In Section \ref{ccoomoo} we specialize to the
case of a concrete model for which $\ca=\carf(\bz^d)$.

\section{the disordered algebra of the observables}
\label{gmlc}

Taking the cue from the concrete model described in the Section
\ref{ccoomoo}, in order to describe disordered models including
Fermions on lattices we list all our assumptions.

We start with a separable $C^{*}$--algebra $\ca$ with an identity
$\idd$, describing the physical observables/fields.\footnote{In the
case of gauge theories,  $\ca$ is obtained as the fixed--point
algebra $\ca=\cf^{G}$ under a pointwise action $\g:g\in
G\mapsto\g_{g}\in\aut(\cf)$ of a field group $G$ (the {\it gauge
group of $1^{st}$ kind}) on another separable $C^{*}$--algebra $\cf$
(the {\it field algebra}). Due to the univalence superselection rule
(cf. \cite{SW}), when we deal with the CAR algebra $\carf(J)$, the
gauge group is precisely $\bz_2$ and the observable algebra is the
even part $\carf(J)_+$. For general theories which include Fermionic
systems, the gauge group includes $\bz_2$, and the even part of the
field algebra is the invariant part under the action of the
generator $\s\in\bz_2\subset G$.} 

We suppose that the spatial translations $\bz^d$ acts in a natural
way on $\ca$ as a group of automorphisms
$\{\a_{x}\}_{x\in\bz^{d}}\subset\aut(\ca)$. 
In addition, we suppose that there exists an automorphism $\s\in\aut(\ca)$ whose square is the identity, commuting with the spatial translations,
\begin{equation}
\label{parsi}
\s^2=\id\,;\quad \s\a_x=a_x\s\,,\quad x\in\bz^d\,.
\end{equation}
In the concrete situation when $\ca=\carf(\bz^d)$, $\a_x$ is the
shift on the lattice of the creators and annihilators by an amount
$x\in\bz^d$, whereas $\s$ is nothing but the parity automorphism
$\Th$. We put
\begin{equation}
\label{fesi}
\ca_{+}:=\frac{1}{2}(e+\s)(\ca)\,,\quad
\ca_{-}:=\frac{1}{2}(e-\s)(\ca)\,.
\end{equation}
Denote by
\begin{equation}
\label{aafff}
\left\{A,B\right\}_\eeps:=AB-\eeps_{A,B}BA
\end{equation}
the {\it graded commutator}, where $\eeps_{A,B}=-1$ if
$A,B\in\ca_{-}$ and $\eeps_{A,B}=1$ in the case of the three
remaining possibilities.

We say that the
$C^{*}$--algebra $\ca$ is {\it graded asypmtotically Abelian} w.r.t.
$\a$, if for each $A,B\in\ca_\pm$,
\begin{equation*}
%\label{aaff}
\lim_{|x|\to+\infty}
\big\{\gpa_{x}(A),B\big\}_\eeps=0\,,
\end{equation*}

In order to introduce the disorder, we consider a standard measure
space $(\Om,\m)$ based on a compact separable space $\Om$, and a
Borel probability measure $\m$. The group $\bz^{d}$ of the spatial
translations is supposed to act on the probability space $(\Om,\m)$
by measure preserving transformations $\{T_{x}\}_{x\in\bz^{d}}$.

A one parameter random group of automorphisms
\begin{equation*}
%\label{5}
(t,\om)\in\br\times\Om\mapsto\t^{\om}_{t}\in\aut(\ca)
\end{equation*}
is acting on $\ca$. It is by definition a representation of $\br$
for each fixed realization  $\om\in\Om$ of the coupling constants.
Furthermore, it is supposed to be jointly measurable in the
$\s$--strong topology. As a consequence of the
Banach--Kuratowski--Pettis Theorem (cf. \cite{Ke}, pag. 211), for
each fixed value $\om\in\Om$, the one parameter group
$t\in\br\mapsto\t^{\om}_{t}\in\aut(\ca)$ is automatically continuous
in the $\s$--strong topology.

Consider, for
$A\in\ca$, the measurable function $f_{A,t}(\om):=\t^{\om}_{t}(A)$. We get
\begin{equation*}
\|f_{A,t}\|_{L^{\infty}(\Om,\m;\ca)}\equiv
\esssup_{\om\in\Om}\|\t_{t}^{\om}(A)\|_{\ca}
=\|A\|_{\ca}\,,
\end{equation*}
where the last equality follows as $\t^{\om}_{t}$ is isometric. As
in the concrete example $\ca=\carf(\bz^d)$ in Section \ref{sec2}, we
further assume that $\t$ acts locally. Namely, if $A\in\ca$, then
the function $f_{A,t}\in L^{\infty}(\Om,\m;\ca)$ belongs to the
$C^{*}$--subalgebra $\ca\otimes L^{\infty}(\Om,\m)$.\footnote{In the
present paper, $\ca\otimes L^{\infty}(\Om,\m)$ means the
$C^*$--algebra obtained by completing the algebraic tensor product
$\ca\odot L^{\infty}(\Om,\m)$ under any $C^*$--cross norm, as any
Abelian $C^{*}$--algebra is nuclear.}

Finally, we assume the following commutation rule
\begin{equation}
\label{4}
\t^{T_{x}\om}_{t}\a_{x}=\a_{x}\t^{\om}_{t}\,,\quad x\in\bz^{d}\,, \om\in\Om\,, t\in\br\,.
\end{equation}
By following the  approach of encoding the disorder in a bigger
algebra(cf. \cite{A1, B, BF, BF1, F, K}), the disordered system
under consideration is described by
$$
\ga:=\ca\otimes L^{\infty}(X,\n)\,.
$$
Notice that, by identifying $\ga$ with a closed subspace of
$L^{\infty}(X,\n;\ca)$, each element $A\in\ga$ is uniquely
represented by a measurable, essentially bounded function
$\om\mapsto A(\om)$ with values in $\ca$. In addition, $\ga$
contains copies $\ca\otimes\idd$ and $\idd\otimes
L^{\infty}(\Om,\m)$ of $\ca$ and $L^{\infty}(\Om,\m)$ respectively,
denoted also by $\ca$ and $L^{\infty}(\Om,\m)$, by an abuse of
notation.

The group $\bz^{d}$ of all the space translations acts naturally on
the $C^{*}$--algebra $\ga$ as
\begin{equation}
\label{trcomm}
\gpa_{x}(A)(\om):=\a_{x}(A(T_{-x}\om))\,, \quad A\in\ga\,,\om\in\Om\,, x\in\bz^d\,.
\end{equation}
Further, as the time translations are supposed to act locally,
\begin{equation*}
%\label{aztemp1}
\gpt_{t}(A)(\om):=\t_{t}^{\om}(A(\om))\,, \quad A\in\ga\,,\om\in\Om\,, t\in\br
\end{equation*}
is a well defined one parameter group of automorphisms of $\ga$, continuous in the
$\s$--strong topology.
In addition, put
\begin{equation*}
%\label{aztemp2}
\gps:=\s\otimes\id_{L^{\infty}(X,\n)}\,.
\end{equation*}
Then the subspaces
$\ga_{+}$ and $\ga_{-}$ are defined as in \eqref{fesi}.

On account of  \eqref{4} and \eqref{trcomm}, it is straightforward
to verify that $\{\gpa_{x}\}_{x\in\bz^{d}}$ and
$\{\gpt_{t}\}_{t\in\br}$ define actions of $\bz^{d}$ and $\br$
respectively on $\ga$ which are mutually commuting. Furthermore, by
\eqref{parsi}, $\gpa_{x}\gps=\gps\gpa_{x}$ for each $x\in\bz^d$.
Concerning the parity of the time translations, we assume that
$\gpt_{t}\gps=\gps\gpt_{t}$ as well. In the concrete cases under
consideration, the parity of the time translations will follow by
the fact that the time translations and the spatial translations are
mutually commuting. Indeed, we have
\begin{prop}
Suppose $\ca=\carf(\bz^d)$. Under all the previous assumptions except the parity for the time evolution, we get $\gpt_{t}\gps=\gps\gpt_{t}$.
\end{prop}
\begin{proof}
If $A,B\in\ga$ with $B(\om)=B$ a constant field, we see that
$\{\gpa_x(A),B\}_\eeps\to0$ when $|x|\to+\infty$. By reasoning as in
Lemma 8.2 of \cite{AM1}, we see that
\begin{equation}
\label{cencaz}
\lim_{|x|\to+\infty}\|[\gpa_{x}(A),B]\|=0
\end{equation}
for each $B\in\ga$, if and only if  $A$ is even. Indeed, let
$A\in\ga$ and $B$ a constant field made by the unitary
$U=a_{x_0}+a_{x_0}^{\dagger}$ for any choice of $x_0\in\bz^d$. We
get
\begin{align*}
&[\gpa_{x}(A),B]=[\gpa_{x}(A_+),B]+[\gpa_{x}(A_-),B]\\
=&[\gpa_{x}(A_+),B]+\{\gpa_x(A_-),B\}
-2\gpa_{x}(A_-)B\,.
\end{align*}
If \eqref{cencaz} holds true, then $\gpa_{x}(A_-)B\to0$ as the first two terms in the l.h.s. go to zero due to CAR.
 Thus $\|A_-\|=\|\gpa_{x}(A_-)\|=\|\gpa_{x}(A_-)B\|\to0$, that is $A=A_+$.
 The converse statement follows from the graded asymptotic
 Abelianness of the CAR algebra $\carf(\bz^d)$ w.r.t. the
 spatial translations. The proposition now follows by applying the
 reasoning in the proof of proposition 8.1 of \cite{AM1} to the time
 translations and spatial translations on the disorder algebra $\ga$.
\end{proof}

\section{states}
\label{sssttt}

Consider a state $\f\in\cs(\ga)$ which is invariant w.r.t. the
spatial translation $\gpa$. Let $(\ch_\f,\pi_\f, U_x, \F)$ be the
GNS covariant quadruple associated to $\f$.

Let $C,D\in\ga$ and $A,B\in\ga_\pm$. We say that the state $\f$ is
{\it graded asymptotically Abelian} w.r.t. $\gpa$ if
\begin{equation}
\label{aaf}
\lim_{|x|\to+\infty}
\f\left(C\big\{\gpa_{x}(A),B\}_\eeps D\right)=0\,,
\end{equation}
where $\{\,{\bf\cdot}\,,B\,{\bf\cdot}\,\}_\eeps$ is the graded commutator given in \eqref{aafff}

The state $\f$ is {\it weakly clustering} w.r.t. $\gpa$ if
\begin{equation}
\label{wcf}
\lim_{N}\frac{1}{|\L_{N}|}\sum_{x\in\L_{N}}
\f(A\gpa_{x}(B))=\f(A)\f(B)\,,
\end{equation}
$\L_{N}$ being the box with a vertex located at the origin,
containing $N^{d}$ points with positive coordinates.\footnote{For
continuous dynamical systems, one uses in \eqref{wcf} the natural
modification $M$ of the Cesaro mean given on bounded measurable
functions, given by
$$
M(f):=\lim_{D\to+\infty}\frac{1}{\vol(\L_{D})}\int_{\L_{D}}f(x)\di^{d}x\,,
$$
$\L_{D}$ being a box with edges of length $D$.} 

The state $\f$ is
{\it $\bz^d$--Abelian} if $E_{\f}\pi_{\f}(\ga)E_{\f}\subset\cb(\ch_\f)$ is
a family of mutually commuting operators, $E_{\f}$ being the
selfadjoint projection onto the invariant vectors for the action of
$U_x$. Furthermore, a state $\f\in\cs(\ga)$ is {\it even} if it is
$\s$--invariant. Denote by $\cs(\ga)_+$ the set of all the even
states.

We report the following result for the
sake of completeness.
\begin{prop}
\label{folk} 
Suppose that $\f\in\cs(\ga)$ is a $\gpa$--invariant,
graded asymptotically Abelian state. Then $\f\in\cs(\ga)_+$ and it
is $\bz^d$--Abelian. In addition, the following assertions are
equivalent.
\begin{itemize}
\item[(i)] $\f$ is  $\gpa$--weakly clustering,
\item[(ii)] $\f$ is  $\gpa$--ergodic.
\end{itemize}
\end{prop}
\begin{proof}
By reasoning as in Example 5.2.21 of \cite{BR2}, we conclude that
$\f$ is automatically even and $\bz^d$--Abelian. Concerning the last
assertion, it is a well--known fact that (i) always implies (ii).
The reverse implication follows as in Proposition 5.4.23 of
\cite{BR2}, the last working also under the weaker condition
\eqref{aaf}.
\end{proof}
In the present paper, the asymptotic Abelianess is always w.r.t. the spatial translations if it is not automatically specified.

Let $\pi$ be a representation of $\ga$. We easily get
$$
\pi(L^{\infty}(\Om,\m))\subset\gz_{\pi}\,.
$$
Suppose that $\pi$ is normal when restricted to
$L^{\infty}(\Om,\m)$. In such a situation, there
exists an essentially unique measurable set $E\subset\Om$, such that
\begin{equation*}
%\label{abcon}
\pi(L^{\infty}(\Om,\m))\sim L^{\infty}(\Om,\n)\,.
\end{equation*}
where $\n$ is nothing but the absolutely continuous measure w.r.t. $\m$ given, for each measurable set $F$, by
$$
\n(F)=\m(F\cap E)\,.
$$
We have also
$$
\pi(\ga)''=\pi(\ca\otimes C(\Om))''\,,
$$
that is $\pi(\ca\otimes C(\Om))$ is a weakly dense separable
$C^{*}$--subalgebra of $\pi(\ga)''$.

We can consider the subcentral decomposition of the restriction of
$\pi$ to the separable $C^{*}$--subalgebra $\ca\otimes C(\Om)$,
w.r.t. $\pi(L^{\infty}(\Om,\m))\equiv\pi(C(\Om))''$, see
\cite{T1}, Theorem IV 8.25. We obtain
\begin{equation}
\label{a1}
\pi=\int^{\oplus}_{\Om}\pi_{\om}\m(\di\om)
\end{equation}
on
$$
\ch_{\pi}=\int^{\oplus}_{\Om}\ch_{\om}\m(\di\om)\,.
$$
The measurable field $\{\pi_{\om}\}_{\om\in\Om}$ of representations
of $\ca\otimes C(\Om)$ is uniquely determined by its restriction to
$\ca$.\footnote{For $\om\in E^{c}$, the complement of the support
$E$ of $\n$, $\pi_{\om}$ will be the trivial representation on the
trivial Hilbert space $\ch_{\om}\equiv\{0\}$.} This follows from the
fact that for each $A\in\ca$ and $f\in L^{\infty}(\Om,\m)$, we have
\begin{equation}
\label{reaa}
\pi(A\otimes f)=\int^{\oplus}_{\Om}f(\om)\pi_{\om}(A\otimes
\idd)\di\m(\om)\,.
\end{equation}
Now by Lemma 8.4.1 of \cite{DX0}, we have
$$
M:=\pi(\ga)''=\int^{\oplus}_{\Om}M_{\om}\n(\di\om)\,,
$$
where for almost all $\om\in\Om$,
$$
M_{\om}=\pi_{\om}(\ca\otimes
L^{\infty}(\Om,\m))''\equiv\pi_{\om}(\ca\otimes\idd)''\,.
$$
As explained in \cite{BF, F}, in order to take into account the disorder, we deal with states
$\f\in\cs(\ga)$ which are normal when restricted to $L^{\infty}(\Om,\m)$ (i.e.
$\f(\idd\otimes f)=\int fg_\f\di\m$ for a uniquely determined $g_\f\in L^{1}(\Om,\m)$). Then
there exists
(\cite{T1}, Proposition IV.8.34) a $*$--weak measurable field
$\{\f_{\om}\}_{\om\in\Om}$ of positive forms on $\ca$ such that,
for each $A\in\ga$,
\begin{equation}
\label{a2}
\f(A)=\int_{\Om}\f_{\om}(A(\om))\m(\di\om)\,,
\end{equation}
the function $\om\mapsto A(\om)$ being the representative of $A$
in $L^{\infty}(\Om,\m;\ca)$.

Consider the GNS representation $\pi_{\f}$ relative to $\f$. It is
straightforward to check that, for almost all $\om\in\Om$,
$\pi_{\f_{\om}}$ is unitarily equivalent to the restriction of
$\pi_{\om}$ to $\ca\otimes\idd\sim\ca$, where $\pi_{\f_{\om}}$ is
the GNS representation of $\f_{\om}$, and $\pi_{\om}$ is the
representation occurring in the decomposition of $\pi_{\f}$ as given
in \eqref{a1}  .
\begin{prop}
\label{001ev} 
Let $\f\in\cs(\ga)$ be such that
$\f\lceil_{L^{\infty}(\Om,\m)}$ is normal. If it is invariant w.r.t.
$\gpa$, then for almost all $\om\in\Om$, the form $\f_{\om}$ in
\eqref{a2} is even.
\end{prop}
\begin{proof}
As $\f$ is invariant w.r.t. the spatial translation, $\f\in\cs(\ga)_+$ (cf. Proposition \ref{folk}). Then
\begin{equation*}
%\label{a3}
\f(A\otimes f)=\int f(\om)\f_{\om}(A)\di\m(\om)=0
\end{equation*}
for each $f\in L^{\infty}(\Om,\m)$ and
$A\in\ca_-$. Thus, for each $A\in\ca_-$ there exists a measurable
set $\Om_{A}\subset\Om$ of full measure such that $\f_{\om}(A)=0$ on
$\Om_{A}$. As $\ca\sim\ca_+\bigoplus\ca_-$ as a Banach space, we can find a countable dense set $\cx\subset\ca_-$. Then
for each $A\in\cx$ we have $\f_{\om}(A)=0$, simultaneously on the measurable set
${\displaystyle\Om_{0}:=
\bigcap_{A\in\cx}\Om_{A}\subset\Om}$ of full measure. Fix
$A\in\ca_-$ and choose sequences $\{A_{n}\}\subset\cx$ converging to $A$. Then we have on
$\Om_0$,
$$
\f_{\om}(A)=\f_{\om}(\lim_{n}A_n)=\lim_{n}\f_{\om}(A_n)=0\,,
$$
that is the positive forms $\f_\om$ given in \eqref{a2} are even almost surely.
\end{proof}
We now show how the  graded asymptotic abelianness of the states on
$\ga$ under consideration directly follows from that of $\ca$.
\begin{prop}
\label{b2}
Let $\f\in\cs(\ga)$ be such that $\f\lceil_{L^{\infty}(\Om,\m)}$ is normal. If $\ca$ is graded asymptotically Abelian then for each $C\in\ga$ and $A,B\in\ga_\pm$ we have,
$$
\lim_{|x|\to+\infty}\f\left(C^*\{\gpa_{x}(A),B\}_\eeps^{*}\{\gpa_{x}(A),B\}_\eeps C\right)
=0\,.
$$
In particular, if $\f$ is invariant w.r.t. $\gpa$, it is graded
asymptotically Abelian.
\end{prop}
\begin{proof}
Let $f,g,h\in L^{\infty}(\Om,\m)$,
and $A,B\in\ca_\pm$, $C\in\ca$. We obtain by \eqref{trcomm}
and \eqref{a2},
\begin{align*}
&\f\left((C\otimes h)^*\{\gpa_{x}(A\otimes f),B\otimes g\}_\eeps^{*}
\{\gpa_{x}(A\otimes f),B\otimes g\}_\eeps C\otimes h\right)\\
=&\int_{\Om}\big|f(T_{-x}\om)g(\om)h(\om)\big|^{2}
\f_{\om}\left(C^*\{\a_{x}(A),B\}_\eeps^{*}\{\a_{x}(A),B\}_\eeps C\right)\di\m(\om)\\
\leq&(\| f\|_{\infty} \|g\|_{\infty}\|h\|_{\infty} \|C\|)^2\|\{\a_x(A),B\}_\eeps\|^2 \longrightarrow 0\,,
\end{align*}
as $\ca$ is graded asymptotically Abelian w.r.t. the spatial
translations. This means
$\|\pi_\f\left(\{\a_x(A),B\}_\eeps\right)\xi\|\to0$ on the Hilbert space of the GNS repersentation of
$\f$, which implies
$$
\lim_{|x|\to+\infty}\f\left(C\{\gpa_{x}(A),B\}_\eeps D\right)
=0
$$
for each $A,B\in\ga_\pm$ and $C,D\in\ga$
\end{proof}
Next we recall the definition of the KMS boundary condition which is
useful for the description of the temperature states of a quantum
dynamical system, see e.g. \cite{BR2}.

A state $\ff$ on the $C^{*}$--algebra $\gb$ satisfies the KMS
boundary condition at inverse temperature $\b$ which we suppose to be
always different from zero,
w.r.t the group of
automorphisms $\{\t_{t}\}_{t\in\br}$ if
\begin{itemize}
\item[(i)] $t\mapsto\ff(A\t_{t}(B))$ is a continuous function for
every $A,B\in\gb$,
\item[(ii)]
$\int\ff(A\t_{t}(B))f(t)\di t=\int\ff(\t_{t}(B)A)f(t+i\b)\di t$
whenever $f\in\widehat{\cd}$, $\cd$ being the
space made of all infinitely
often differentiable compactly supported functions in $\br$.
\end{itemize}
For the equivalent characterizations of the KMS boundary
condition, the main results about KMS states, and finally the
connections with Tomita theory of von Neumann algebras, see
e. g. \cite{BR1, BR2, St} and the references cited therein.

It is well known that the cyclic vector $\Om_{\ff}$ of the GNS
representation $\pi_{\ff}$ is also separating for
$\pi_{\ff}(\gb)''$. Denote by $\s^{\ff}$ its modular group.
According to the definition of KMS boundary condition, we have
\begin{equation}
\label{modgns}
\s^{\ff}_{t}\circ\pi_{\ff}=\pi_{\ff}\circ\t_{-\b t}\,.
\end{equation}
We end the present section by listing some useful properties of
states, which are normal (not necessarily KMS) when restricted to
$L^{\infty}(\Om,\m)$, contained in Section 3 and 4 of \cite{BF}.
\begin{prop}
\label{001}
Let $\f$ be $\gpt$--KMS state on
$\ga$ at inverse temperature $\b$ which is normal when restricted to $L^{\infty}(\Om,\m)$.

Then, for almost all $\om\in\Om$, the forms $\f_{\om}$ given in \eqref{a2}
are $\t^{\om}$--KMS at the same inverse temperature $\b$.
\end{prop}
\begin{proof}
The proof is along the same lines as those of Proposition
\ref{001ev}, which is reported in Section 3 of \cite{BF}.
\end{proof}
\begin{thm}
\label{c1}
Let $\f$ be an invariant state on $\ga$ which is normal if restricted to $L^{\infty}(\Om,\m)$.
Consider the decomposition appearing in \eqref{a2}. Then
\begin{itemize}
\item[(i)] $\f_{\om}\circ\a_{x}=\f_{T_{-x}\om}$ for all $x\in\bz_{d}$,
\item[(ii)] $\gz_{\pi_{\f_{\om}}}\cong\gz_{\pi_{\f_{T_{x}\om}}}$
for all $x\in\bz_{d}$,
\end{itemize}
where the above equalities, as well as the unitary equivalence,
are satisfied almost everywhere.

In addition, if the action $T_{x}$ on $\Om$ is ergodic, then
$\f(\idd)=1$ almost surely.
\end{thm}
Consider for $\a\in\{\infty\}\cup\{1,2,\dots\}\cup
\{\l_{\infty},\l_{0},\l_{1},\dots\}$, the Abelian von Neumann
algebras $L^{\infty}(E_{\a},\n_{\a})$ defined as follows. For
$\a\in\{\infty\}\cup\{1,2,\dots\}$, $(E_{n},\n_{n})$ is the
countable set $E_{n}={\bf n}$ of cardinality $n$, equipped with
the counting measure $\n_{n}$ (the symbol $\infty$ corresponds to
the denumerable cardinality). For $\a=\l_{n}$,
$(E_{\l_{n}},\n_{\l_{n}})$ is the disjont union $[0,1]\cup{\bf n}$
equipped with the measure $\n_{\l_{n}}$ made of the Lebesgue
measure $\l$ on $[0,1]$, and the counting measure on ${\bf n}$
(the value $n=0$ corresponds to $L^{\infty}([0,1],\l)$).

As a corollary to Theorem \ref{c1}, we have
\begin{cor}
\label{c2} 
Let $\f$ be as in Theorem \ref{c1}, and suppose that $T_{x}$ acts ergodically on
$\Om$. Then there exists a
unique $\a\in\{\infty\}\cup\{1,2,\dots\}\cup
\{\l_{\infty},\l_{0},\l_{1},\dots\}$ such that
$$
\gz_{\pi_{\f_{\om}}}\sim L^{\infty}(E_{\a},\n_{\a})
$$
almost surely.
\end{cor}
We cannot conclude that $\gz_{\pi_{\f_{\om}}}$ is almost surely of
a unique multiplicity class.\footnote{Notice that there are
uncountable many Abelian von Neumann algebras acting on separable
Hilbert spaces, up to unitary equivalence.} However, for locally
normal invariant KMS states we have that $\gz_{\pi_{\f_{\om}}}$ is
almost surely of infinite multiplicity, see below.
\begin{defin}
For the $\bz_2$--graded models considered in the present paper, we denote by
$\cs_{N}(\ga)\subset\cs(\ga)_+$ the subset of the even states $\f$ such that
$\f\lceil_{L^{\infty}(\Om,\m)}$ is normal. The set $\cs_{NI}(\ga)\subset\cs_{N}(\ga)$ consists of those states $\f\in\cs_{N}(\ga)$ which are in addition $\gpa$--invariant.
\end{defin}
In the classical setting, the the measurable field $\{\f_\om\}$ arising from the direct integral decomposition of a temperature state $\f\in\cs_{NI}(\ga)$ is called an {\it Aizenman--Wehr metatstate}, see \cite{AW}. In the quantum case, the counterparts of the Aizenman--Wehr metatstates were naturally considered early in \cite{K}.

\section{spectral properties of $\bz_2$--graded asymptotically abelian dynamical systems}
\label{asgab}

The present section is devoted to prove some useful results
concerning the spectral properties of  dynamical systems which are
$\bz_2$--graded asymptotically Abelian, the last being the natural setting
for theories including Fermi particles. The results proved below and in the next section 
have a self contained interest as they provide the generalization to
the $\bz_2$--graded dynamical systems of the pivotal results of
\cite{A, HL, Se}, and those reviewed in \cite{L} for the natural
applications to the investigation of the structure of the local
algebras in Quantum Field Theory. For the definition and the main
properties of the Borchers, Connes and Arveson spectra $\G_B(\a)$,
$\G(\a)$, $\sp(\a)$ of an action $\a$, and then the Borchers and
Connes invariants $\G_B(M)$, $\G(M)$ of a von Neumann algebra $M$,
respectively, the reader is referred to the original papers
\cite{Bo, C} and the books \cite{P, Su}.

One of the main objects of interest in the investigation of the
spectral properties of non commutative dynamical systems is the
Arveson spectrum and its connection with the spectrum of the group
of unitaries implementing the dynamics in the covariant GNS
representation, see e.g. \cite{P}. It was shown that for the
dynamical systems based on random interactions treated in \cite{B,
BF} the Arveson spectrum is almost surely independent of the
disorder.\footnote{Compare with the analogous result (\cite{KS},
Th\'eor\`eme III.1) concerning the spectrum of a one dimensional
random discretized Schr\"odinger operator.} In addition, for most of
the KMS states considered in \cite{B}, the spectrum of the
associated modular group was also found to be independent of the
disorder. The proofs of such results depend mainly on the general
properties assumed for our disordered model and not so much on the
local structure of the $C^*$--algebra $\ca$. Therefore the proofs of
these results which appear as Theorem 5.3 and Proposition 5.5 in
\cite{B} can be reproduced {\it mutatis mutandis} for the situation
under consideration in the present paper. To this end, we suppose that $\ca$ is a separable unital $C^*$--algebra. Let $\t^{\om}_{t}$ be jointly
measurabile in $t$, $\om$. Assume the commutation rule
\eqref{4} and the ergodicity of the action $T_x$ of the spatial
translations on the sample space $(\Om,\m)$. Put
$\ga:=\ca\otimes L^\infty(\Om,\m)$ and choose a KMS
state $\f\in\cs_{NI}(\ga)$ at inverse temperature $\b\neq0$. Consider the forms $\f_\om$
in \eqref{a2}, which thanks to Propositions \ref{001ev} and
\ref{001}, are even and satisfy the KMS condition at the same
inverse temperature $\b$ almost surely.
\begin{thm}
\label{ba}
Under the above assumptions,
there exists a measurable
set $F\subset\Om$ of full measure, and a closed set $\S\subset\br$
such that $\om\in F$ implies $\sp(\t^{\om})=\S$.

In addition, if $\ca$ is simple, then
\begin{equation*}
%\label{prrr3}
\b\sp(\t^{\om})=-\s(\ln\D_{\f_{\om}})
\end{equation*}
almost surely, where $\D_{\f_{\om}}$ is the modular operator associated to $\f_{\om}$.
\end{thm}
\begin{proof}
We get by \cite{P}, Proposition 8.1.9,
$$
\sp(\t^{\om})=\bigcap_{f\in L^{1}(\br)}
\big\{s\in\br\;\big|\;|\hat{f}(s)|\leq\|\t^{\om}_{f}\|\big\}
$$
where `` $\hat{}$ '' stands for (inverse) Fourier transform, and
$$
\t^{\om}_{f}(A):=\int_{-\infty}^{+\infty}f(t)\t^{\om}_{t}(A)\di
t\,,
$$
the integral being understood in the Bochner sense. By a standard
density argument, we can reduce the situation to a dense set
$\{f_{k}\}_{k\in\bn}\subset L^{1}(\br)$. Define
$\G_{k}(\om):=\|\t^{\om}_{f_{k}}\|$. It was shown in \cite{B} that
the functions $\G_{k}$ are measurable and invariant. By
ergodicity, they are constant almost everywhere. Let
$\{N_{k}\}_{k\in\bn}$ be null subsets of $\Om$ such that, for each
$k\in\bn$ and $\om\in N_{k}^{{}^{c}}$,
$$
\G_{k}(\om)=\|\G_{k}\|_{\infty}\,.
$$
Consider $F:=\big(\bigcup_{k\in\bn}N_{k}\big)^{c}$, and take
$\S:=\sp(\t^{\om_{0}})$, where $\om_{0}$ is any element of $F$. As
an immediate consequence of this, we have that $F$ is a measurable
set of full measure, and $\om\in F$ implies $\sp(\t^{\om})=\S$.

Consider the GNS covariant representation $(\ch_\om,\pi_\om, U_\om, \F_\om)$ of $\f_\om$.
Thanks to the facts that, on a measurable set $F\in\Om$ of full measure, $\F_\om$ is a standard vector for $\pi_\om(\ca)''$ and $\pi_\om$ is faithful as $\ca$ is simple, we get
for $f\in L^1(\br)$, $U_\om(f):=\int f(t)U_\om(t)\di t=0$
if and only if $\t^\om_f\equiv\int f(t)\t^\om_t\di t=0$. For $\om\in F$, this leads to
$\sp(\t^{\om})=-\frac1{\b}\s(\ln\D_{\f_{\om}})$ by \eqref{modgns}.\footnote{Notice that if $\ca$ is not simple, we merely have $\sp(\tilde\t^{\om})=-\frac1{\b}\s(\ln\D_{\f_{\om}})$ where
$\tilde\t^{\om}_t:=\ad U_\om(t)$ acting on $\pi_{\f_\om}(\ca)''$. It can be showed as before, that it is independent on the disorder.}
\end{proof}
Now we generalize a standard result (cf. \cite{A, HL, Se})
on the spectrum of the modular action of asymptotically Abelian
systems to the $\bz_2$--graded case. The proof in the graded case will be more involved then that for the asymptotically Abelian situation.
\begin{thm}
\label{lng}
Let $(M, G, \t)$, $(M, H, \a)$ be $W^*$--dynamical systems based on the $\bz_2$--graded
$W^*$--algebra $M$, with $G$ locally compact and Abelian. Suppose that the actions $\t$ and
$\a$ are even, commute each other, and leave invariant the faithful normal state $\f$.

If for an invariant mean
$m_H$ on $H$,
and for each $A\in M_\pm$, $B\in Z(M^\t)_\pm$,
\begin{equation}
\label{czlo}
m_H\left\{\f\left(\{\a_{h}(A),B\}_\eeps^{*}\{\a_{h}(A),B\}_\eeps \right)\right\}=0\,,
\end{equation}
then $\G_B(\t)=\sp(\t)$.\footnote{For the definition of invariant means of a (semi)group, we refer the reader to Section 17 of \cite{HR}.}
%$\G_B(\t)$, $\sp(\t)$ being the Borchers and the Arveson spectrum of $\t$, respectively.
\end{thm}
\begin{proof}
Fix $E\in Z(M^\t)$ with central support (in $M$) $c(E)=\idd$. We
notice that $\s(E)\in Z(M^\t)$ as  $\t$ is even. In addition,
$c(\s(E))=\idd$ too. Let $p\in\sp(\t)$ and fix a closed neighborhood
$V$ of $p$. Then there exists a nonzero element $A\in M(\t,V)$, the
last being the spectral subspace associated with the closed subset
$V\subset \hat G$ (cf. \cite{P}). If $M(\t,V)\bigcap M_+\neq\{0\}$,
we argue as in Theorem 2 of \cite{HL}, that there exists $h\in H$
such that $E\a_h(A)E\neq0$ where $A \in M(\t,V)\bigcap M_+$. In this
case $E\a_h(A)E\in M_E(\t^E,V)$, where $\t^E$ is the restricted
action of $\t$ on the reduced algebra $M_E$. If for some closed
neighborhood $V$ of $p$, $M(\t,V)\subset M_-$, we proceeds as
follows. Namely, fix a nonzero element $A\in M(\t,V)$ and suppose
that $E\a_h(A)E\neq0$ for some  $h\in H$. Then we conclude
that$E\a_h(A)E\in M_E(\t^E,V)$, where $\t^E$ is the restricted
action of $\t$ on the reduced algebra $M_E$. If on the other hand,
$E\a_h(A)E=0$ for each $h\in H$, then we get,
\begin{align}
\label{czlo1}
&E\s(E)\a_h(A^*A)\s(E)E=\s(E)E\a_h(A^*)E\a_h(A)E\s(E)\\
+E\s(E)&\a_h(A^*)\{\a_{h}(A),\s(E)\}_\eeps E=E\s(E)\a_h(A^*)\{\a_{h}(A),\s(E)\}_\eeps E\,.\nn
\end{align}
Let now $\ce_H:M\to M^\a$ (resp. $\ce_G:M\to M^\t$) be the normal faithful conditional expectation onto
the fixed point subalgebra $M^\a$ (resp. $M^\t$) leaving invariant the state $\f$, which exists by the Kovacs--Sz\"ucs Theorem (see e.g. \cite{BR1}, Proposition 4.3.8). We have by
Cauchy--Schwarz Inequality, Holder Inequality, \eqref{czlo} and
\eqref{czlo1},
\begin{align*}
%\label{czlo1}
&\f(E\s(E)\ce_H(A^*A)\s(E)E)=m_H\left\{\f(E\s(E)\a_h(A^*A)\s(E)E)\right\}\\
=&m_H\left\{\f(E\s(E)\a_h(A^*)\{\a_{h}(A),\s(E)\}_\eeps E)\right\}\\
\leq&\|A\|m_H\left\{\f(\{\a_{h}(A),\s(E)\}^*_\eeps\{\a_{h}(A),\s(E)\}^*_\eeps)\right\}^{1/2}=0\,.
\end{align*}
This implies that $E\s(E)\ce_H(A^*A)\s(E)E=0$ as $\f$ is faithful. Now,
$$
\ce_G(\s(E)\ce_H(A^*A)\s(E))E=\ce_G(E\s(E)\ce_H(A^*A)\s(E)E)=0
$$
as $E\in Z(M^\t)$. In addition, $\ce_G(\s(E)\ce_H(A^*A)\s(E))=0$ as
$c(E)=\idd$, which implies $\s(E)\ce_H(A^*A)\s(E)=0$ as $\ce_G$ is
faithful. By repeating the same argument  for
$\s(E)\ce_H(A^*A)\s(E)$ we get $\ce_H(A^*A)=0$, which implies the
contradiction $A=0$ as $\ce_H$ is faithful as well. Thus, either
when $M(\t,V)\bigcap M_+\neq\{0\}$ or $M(\t,V)\subset M_-$, if
$p\in\sp(\t)$ and $V$ is any closed neighborhood $V$ of $p$, there
always exists $A\in M$ such that $E\a_h(A)E\neq0$ and $E\a_h(A)E\in
M_E(\t^E,V)$. This means that $p\in\sp(\t^E)$ as well. As by
Proposition 1 of \cite{HL},
$$
\G_B(\t)=\bigcap_{\{E\in Z(M^\t)\mid c(E)=\idd\}} \sp(\t^E)\,,
$$
this leads to the assertion.
\end{proof}

\section{the type of von Neumann algebras associated to graded asymptotically abelian dynamical systems and Fermionic disordered models}
\label{dsaspt}

The present section is devoted investigate the type of von Neumann algebras associated to 
$\bz_2$--graded asymptotically Abelian dynamical systems. The natural application will concern the von Neumann algebras generated by temperature states of disordered Fermionic models. 

We start by proving that a von Neumann algebra with a nontrival
infinite semifinite summand cannot carry an action which is graded
asymptotically Abelian w.r.t. the strong operator topology. The proof for the
cases which include Fermionic systems is more involved than the
original one in \cite{D, L}. As usual $M$ will be a $\bz_2$--graded
von Neumann algebra whose grading is generated by an automorphism
$\s\in\aut(M)$ with $\s^2=\id$.
\begin{lem}
\label{ppprrr}
Let $M$ be a $\bz_2$--graded semifinite von Neumann algebra with a normal semifinite faithful trace $\t$.
Then there exists a selfadjoint projection $E\in M_+$ with $0<\t(E)<+\infty$.
\end{lem}
\begin{proof}
Choose a selfadjoint projection $F\in M$ with $0<\t(F)<+\infty$, and
consider the splitting $F=F_++F_-$ of $F$ into even and odd parts.
As $F$ is a selfadjoint projection, we get $F_+=F_+^*F_++F_-^*F_-$.
This leads to $0<\t(F_+)<+\infty$, otherwise
$\t(F_+^*F_+)+\t(F_-^*F_-)=0$ which implies $F_+=F_-=0$. Let
$F_+=\int\l\di E(\l)$
be the resolution of the identity of $F_+$. By considering first the
approximation of continuous functions by polynomials in the uniform
topology, and then the pointwise approximation of Borel functions
with uniformly bounded continuous ones, we see that the projections
$E(\l)$ are even. In addition, the increasing function
$\l\mapsto\t(E(\l)))$ is the cumulative function associated with a
Borel measure on the interval $[0,1]$ such that
$$
\t(F_+)=\int\l\di\t(E(\l))\,.
$$
Then there exists $\l_0\in[0,1]$ such that $0<\t(E(\l_0))<+\infty$.
The projection we are looking for is $E:=E(\l_0)\in M_+$.
\end{proof}
Consider on $\bn$ a mean $m$ concentrated at the infinity of $\bn$.\footnote{The mean $m$ concentrated at infinity is uniquely determined by a state on the Corona Algebra. Namely,
$m\in\cs\left(B(\bn)/B_0(\bn)\right)$, where $B(\bn)$ is the $C^*$--algebra of all the bounded functions on
$\bn$, and $B_0(\bn)$ is made of those bounded functions vanishing at infinity.}
\begin{thm}
\label{btra}
Let $M$ be a $\bz_2$--graded von Neumann algebra. If for a mean $m$ concentrated at the infinity of $\bn$, and for a sequence $\{\a_n\}_{n\in\bn}\subset\aut(M)$ of even automorphisms,
\begin{equation}
\label{prrr1}
m\left\{\f([\a_n(A),B]^*[\a_n(A),B])\right\}=0\,,
\end{equation}
for each $A\in M$, $B\in M_+$ and $\f\in\cs(M)\bigcap M_*$,
then the properly infinite semifinite summand in $M$ is trivial.
\end{thm}
\begin{proof}
After using Lemma \ref{ppprrr}, the proof proceeds along the same
lines as those of the analogous results in \cite{D, L}. Indeed, let
$F$, $G$, $H$ be the central projections of $M$ corresponding to the
finite, infinite semifinite, and purely infinite part of $M$,
respectively. We have that all of them are invariant under the
action of $\s$ and the $\a_n$, otherwise they would not be maximal.
Thus, we can suppose that $M$ is infinite semifinite. Choose a
normal semifinite faithful trace $\t$ (cf. \cite{T1}, Theorem V.2.15)
and a even selfadjoint projection $E\in M_+$ with $\t(E)=1$, which
exists by Lemma \ref{ppprrr}. Consider on $M$ the state
$$
\f(A):=\t(EAE)\,,\quad A\in M\,.
$$
Define $\psi\in\cs(M)$ as
$$
\psi(A):=m\left\{\f\circ\a_n(A)\right\}\,,\quad A\in M\,.
$$
By taking into account \eqref{prrr1}, we have by the Cauchy--Schwarz Inequality (cf. \cite{T1},  Proposition I.9.5), and Holder inequality,
\begin{align}
\label{prrr2}
&\left|m\left\{\f\left(\a_{n}(A)[\a_{n}(B),E]\right)\right\}\right|\\
\leq\|A\|m&\left\{\f\left([\a_{n}(B),E]^*[\a_{n}(B),E]\right)\right\}^{1/2}=0\,.\nn
\end{align}
Thanks to \eqref{prrr2}, we compute
\begin{align*}
\psi(A&B)=m\left\{\t\left(E\a_{n}(A)\a_{n}(B)E\right)\right\}
=m\left\{\f\left(\a_{n}(A)[\a_{n}(B),E]\right)\right\}\\
+&m\left\{\t\left(E\a_{n}(A)E\a_{n}(B)E\right)\right\}
=m\left\{\t\left(E\a_{n}(B)E\a_{n}(A)E\right)\right\}\\
=&m\left\{\f\left(\a_{n}(B)[\a_{n}(A),E]\right)\right\}
+m\left\{\t\left(E\a_{n}(B)E\a_{n}(A)E\right)\right\}\\
=&m\left\{\t\left(E\a_{n}(B)\a_{n}(A)E\right)\right\}=\psi(BA)\,.
\end{align*}
Namely, $\psi$ is a (possibly non normal) tracial state on $M$ which is a contradiction.\footnote{Choose selfadjoint mutually orthogonal projections $E_j\in M$, $j=1,2$ equivalent to the identity $\idd$ such that
$E_1+E_2=\idd$, which always exist as $M$ is properly infinite. Then
$1=\psi(\idd)=\psi(E_1+E_2)=\psi(E_1)+\psi(E_2)=\psi(\idd)+\psi(\idd)=2$.}
\end{proof}

Consider a sequence $\{\a_{n}\}_{n\in\bn}$ of even $*$--automorphisms
of a $\bz^2$--graded $C^{*}$--algebra $\gb$, together with an even state $\f$ on $\gb$
which is invariant for $\{\a_{n}\}$. Let $\{U_{n}\}$, $V$ be the
covariant implementation of $\{\a_{n}\}$, and of the grading $\s$ relative to the GNS
triplet $(\pi_{\f},\ch_{\f},\Psi_{\f})$ corresponding to $\f$, respectively.
Denote by $\tilde\a_{n}:=\text{ad}U_{n}$,  $\tilde\s:=\text{ad}V$ the corresponding
automorphisms on $M:=\pi_{\f}(\gb)''$.
\begin{lem}
\label{modst} 
Suppose that $\Psi_{\f}$ is separating for
$\pi_{\f}(\gb)''$.\footnote{Such a condition is equivalent to the
fact that the state $(\,\cdot\,\Psi_{\f},\Psi_{\f})$ is a KMS state
(at inverse temperature 1) on $\pi_{\f}(\gb)''$ w.r.t. the modular
automorphism group constructed from $\Psi_{\f}$. In addition, the
previous conditions are also equivalent to the fact that the support
$s(\f)\in\gb^{**}$ of the state $\f$ in the bidual $\gb^{**}$ is
central, see e.g. \cite{SZ}, Section 10.17.} Then
\begin{equation*}
%\label{refe}
\lim_{n}\f\left(\{\a_{n}(A),B\}_\eeps^{*}\{\a_{n}(A),B\}_\eeps \right)=0
\end{equation*}
for every $A,B\in\gb$ implies
$$
\lim_{n}\{\tilde\a_{n}(X),Y\}_\eeps\xi
$$
for every $X,Y\in\pi_{\f}(\gb)''$ and $\xi\in\ch_{\f}$.
\end{lem}
\begin{proof}
As $\Psi\equiv\Psi_{\f}$ is cyclic for $M'$, it is enough to show
that $n\to\infty$ implies $\{\tilde\a_{n}(X),Y\}_\eeps\Psi\to 0$.

Let $\eps>0$ and $X,Y\in M_{1}$ be fixed. Then there exist
$X',Y'\in M'\backslash\{0\}$ such that
$$
\|(X-X')\Psi\|<\eps,\quad \|(Y-Y')\Psi\|<\eps\,.
$$
Let $X,Y\in M_\pm$, we can find  $A,B\in\gb_\pm$
with $\|\pi_{\f}(A)\|\leq1$, $\|\pi_{\f}(B)\|\leq1$, such that
\begin{align*}
&\|(X-\pi_{\f}(A))\Psi\|<\left(1\wedge1/\|Y'\|\right)\eps\, ,\\
&\|(Y-\pi_{\f}(B))\Psi\|<\left(1\wedge1/\|X'\|\right)\eps\, .
\end{align*}
Indeed, for simplicity let $X\in M_-$ (the situation $X\in M_+$
follows analogously). Let $P$ be the projection of $M$ onto $M_-$.
By the Kaplansky Density Theorem, there exists $C\in\gb$ such that
$$
\|(X-\pi_{\f}(C))\Psi\|<\left(1\wedge1/\|Y'\|\right)\eps\,.
$$
Take $B:=\frac{C-\s(C)}2$. By our assumptions, $X=P(X)$ and $\pi_{\f}(B)=P(\pi_{\f}(C))$. By taking into account the last, we get
\begin{align*}
&\|(X-\pi_{\f}(B))\Psi\|=\|(P(X-\pi_{\f}(B)))\Psi\|\\
=&\frac12\|(X-\pi_{\f}(C))\Psi-V(X-\pi_{\f}(C))\Psi\|\\
\leq&\|(X-\pi_{\f}(C))\Psi\|<\left(1\wedge1/\|Y'\|\right)\eps\,.
\end{align*}
We treat the situation $X,Y\in M_-$, the other cases being analogous.
\begin{align*}
\|(\{\tilde\a_{n}(X),Y\}-\pi_{\f}(\{\tilde\a_{n}(A),B\}))\Psi\|
\leq&\|(\tilde\a_{n}(X)Y-\tilde\a_{n}(\pi_{\f}(A))\pi_{\f}(B))\Psi\|\\
+&\|(Y\tilde\a_{n}(X)-\pi_{\f}(B)\tilde\a_{n}(\pi_{\f}(A)))\Psi\|\,.
\end{align*}
As both the terms of the r.h.s. of the above inequality is
estimated in the same way, we consider only the first one. We get
\begin{align*}
&\|\left(X\tilde\a_{n}(Y)-\pi_{\f}(A)\tilde\a_{n}(\pi_{\f}(B))\right)\Psi\|
\leq\|\tilde\a_{n}(X-\pi_{\f}(A))(Y-Y')\Psi\|\\
+&\|Y'U_n(X-\pi_{\f}(A))\Psi\|
+\|\pi_{\f}(\a_{n}(A))(Y-\pi_{\f}(B))\Psi\|\\
<&2\eps+\|Y'\|\left(1\wedge1/\|Y'\|\right)\eps+\left(1\wedge1/\|X'\|\right)\eps
\leq4\eps
\end{align*}
which leads to the assertion.
\end{proof}
As a direct consequence, we have the following result describing the
structure of the von Neumann algebras generated by GNS
representations associated with  a $\bz_2$--graded asymptotically
Abelian state
%on $C^{*}$--algebras
such that its support in the bidual is central.
Such a result can be applied
immediately to the model under consideration, and yet wider
applications are possible.
\begin{thm}
\label{reff}
Let $(\gb,\a,\f)$ be a $C^*$--dynamical system, with $\gb$ a $\bz_2$--graded
$C^*$--algebra, $\a$ an even action of $\bz^d$, and finally $\f$ an even state which is invariant under the action of $\a$.
Suppose that $\ch_{\f}$ is a separable Hilbert space and $\Psi_{\f}$ is separating for
$\pi_{\f}(\gb)''$, $(\pi_{\f},\ch_{\f},\Psi_{\f})$ being the GNS triplet relative to $\f$.
If
$$
\lim_{|x|\to+\infty}\f\left(\{\a_{x}(A),B\}_\eeps^{*}\{\a_{x}(A),B\}_\eeps \right)=0\,,
$$
then
$\pi_{\f}(\gb)''$ does not contain type $\ty{I_{\infty}}$,
$\ty{II_{\infty}}$ and $\ty{III_{0}}$ components.
\end{thm}
\begin{proof}
Let $\tilde\f:=\langle\,{\bf\cdot}\,\Psi_\f,\Psi_\f\rangle$. We start by noticing that
$\G_B(M)=\G_B(\s^{\tilde\f})$ (cf. \cite{HL}, Proposition 1) as the last does not depend on the faithful state $\tilde\f$ on
$M$. By Proposition C1 of \cite{BF}, $\e(\G_B(M))=\rm{S}(M)\backslash\{0\}$, $\rm{S}(M)$ being Connes $\rm{S}$--invariant (cf. \cite{C}). Finally, $\sp(\s^{\tilde\f})=\ln\s(\D_{\tilde\f})$.

Let $E$ be the central projection corresponding to the type
$\ty{III_{0}}$ component of $\pi_{\f}(\gb)''\equiv M$ which is
well--defined as $M$ is acting on a separable Hilbert space, see
\cite{S1}. Assume $E>0$. As for $x\in\bz^d$, $\tilde\a_x(E)=E$, and
$\tilde\s(E)=E$, we can suppose that $E=\idd$, that is $M$ is itself
of type $\ty{III_{0}}$. As $M$ is supposed of type $\ty{III_{0}}$,
we get $\rm{S}(M)=\{0,1\}$. By considering the Cesaro mean as that
described in \eqref{wcf}, we obtain by Lemma \ref{modst} and Theorem
\ref{lng}, $\G_B(M)=\sp(\s^{\tilde\f})$. It readily follows from
these results that,
\begin{align*}
&\s(\D_{\tilde\f})\backslash\{0\}=\e(\sp(\s^{\tilde\f}))=\e(\G_B(\s^{\tilde\f}))\\
=&\e(\G_B(M))=\rm{S}(M)\backslash\{0\}=\{0,1\}\backslash\{0\}\,.
\end{align*}
This means that $\s(\D_{\tilde\f})=\{1\}$ as $0$ cannot be an
isolated point of the spectrum. Thus we have arrived at the
contradiction that $\tilde\f$ is a trace. Hence, $M$ cannot contain
the type $\ty{III_{0}}$ component. The proof follows as the infinite
semifinite part is avoided by the application of Theorem \ref{btra},
taking into account Lemma \ref{modst}.
\end{proof}

As a direct consequence of the previous results on the spectral properties, we show that the temperature states of the disordered systems under consideration can generate only type $\ty{III}$ von Neumann algebras, except the type $\ty{III_0}$.
\begin{lem}
\label{simfind}
Let $\gb$ an infinite dimensional simple separable $C^*$--algebra together with its representation $\pi$.
Then $\pi(\gb)''$ does not contain the type $\ty{I_{\mathop{fin}}}$ component.
\end{lem}
\begin{proof}
Consider the central projection $E_n$ relative to the $\ty{I_{n}}$
component, $n\in\bn$, of $\pi(\gb)''$, which exists by \cite{N}. By
considering the representation $\pi_n:=\pi(\,{\bf\cdot}\,)E_n$, we
can assume that $\pi$ itself contains only the type  $\ty{I_{n}}$
component. Consider the direct integral decomposition
$\pi=\int\pi_x\di\n$ of $\pi$ w.r.t. the center of $\pi(\gb)''$. We
get that $\pi_x(\gb)''$ is isomorphic to the full matrix algebra
$\bm_n(\bc)$, $\n$--almost surely. But this is impossible as $\pi_x$
is faithful almost surely as $\gb$ is simple, and then
$\pi_x(\gb)''$ cannot be $\bm_n(\bc)$ as $\gb$ is infinite
dimensional. The proof follows as $n$ is arbitrary and
$E_{\text{fin}}=\bigoplus_{n\in\bn}E_n$, $E_{\text{fin}}$ being the
central projection corresponding to the finite component of
$\pi(\gb)''$.
\end{proof}
\begin{thm}
\label{c6}
Let $\ga=\ca\otimes L^{\infty}(\Om,\m)$, and $\f\in\cs_{NI}(\ga)$ be a KMS state 
 at inverse temperature $\b\neq0$ w.r.t. the time evolution which is supposed to be nontrivial. Suppose that  $\ca$ is separable, simple and graded asymptotically Abelian.
Then only type $\ty{III_{\l}}$ factors, $\l\in(0,1]$, can
appear in its central decomposition.

If in addition,
$\gz_{\pi_{\f}}\sim L^{\infty}(\Om,\m)$, then there exists a
unique $\l\in(0,1]$ such that $\pi_{\f_{\om}}(\ca)''$ are type
$\ty{III_{\l}}$ factors almost surely.
\end{thm}
\begin{proof}
Let $\pi=\int^{\oplus}_{\Om}\pi_{\om}\di\m(\om)$ be the direct
integral decomposition of $\pi$ as explained in Section
\ref{sssttt}. By taking into account \eqref{reaa}, $\pi_{\om}$ is
indeed a representation of $\ca$, and
$\pi_{\om}(\ga)''=\pi_{\om}(\ca)''$ almost surely. By Lemma
\ref{simfind}, we conclude that $\pi_{\om}(\ga)''$ cannot contain
the type $\ty{I_{\text{fin}}}$ component, almost surely. This means
that $\pi(\ga)''$ does not contain the type $\ty{I_{\text{fin}}}$
component. By Proposition \ref{b2} and Theorem \ref{reff}, the type
$\ty{I_{\infty}}$, $\ty{II_{\infty}}$ and $\ty{III_{0}}$ are also
absent.

Concerning the type $\ty{II_{1}}$ component, let
$E\in\gz_{\pi_{\f}}$ be the corresponding central projection which
we assume to be non zero. By proposition 3.1 of \cite{BF}, the state
$$
\f_E:=\frac{\langle\pi_\f(\,{\bf\cdot}\,)E\F,\F\rangle}{\langle E\F,\F\rangle}
$$
is a KMS state which is normal w.r.t. $\f$. This means that
$\f_E\lceil_{L^{\infty}(\Om,\m)}$ is normal. In addition,
$\f_E\in\cs_{NI}(\ga)$ as $VEV^*=E$ for each
$V\in\cn(\pi_{\f_\om}(\ga)'')$, $\cn(\pi_{\om}(\ga)'')$ being the
normalizer of $\pi_{\f}(\ga)''$ in $\cb(\ch_\f)$. Thus, we assume
without loss of  generality that $\pi_{\f}(\ga)''$ is a type
$\ty{II_{1}}$ von Neumann algebra. Denote as usual
$\tilde\t^\om_t:=\ad U_\om(t)$ and  $\tilde\gpt_t:=\ad U(t)$ on
$\pi_{\f_\om}(\ga)''$ and  $\pi_\f(\ga)''$, respectively. By
applying Proposition C2 of \cite{BF}, theorems \ref{ba} and
\ref{lng}, we get
$$
\sp(\gpt)=\sp(\t^\om)=\sp(\tilde\t^\om)=\sp(\tilde\gpt)=\G_B(\tilde\gpt)=\G_B(\pi_\f(\ga)'')=0\,,
$$
where the first three equalities hold true almost surely. But this is a contradiction as $\gpt$ is supposed to be non trivial, see e.g. \cite{Su}, propositions 3.2.8 and 3.2.9.

If $\gz_{\pi_{\f}}\sim L^{\infty}(\Om,\m)$, then
$M_{\om}\equiv\pi_{\f_{\om}}(\ca)''$ are factors almost surely. This
implies $\G_{B}(M)=\G(M_{\om})$ almost surely, where $\G$ is Connes
$\G$--invariant (\cite{C}). The theorem follows from the previous
part.
\end{proof}

\section{a concrete disordered fermionic model}
\label{ccoomoo}

We apply the previous results which are quite general in nature to a
pivotal model. In fact,
this model can be viewed as a disorderd spinless Hubbard model, provided the common distribution of the coupling constants is one--sides. The most interesting situation will be when such a common distribution is two--sides.
The associated nearest neighbor random Hamiltonian of this Fermionic
system we have in mind has the form
\begin{equation}
\label{ffham}
H=\sum_{\{(x,y)\in\bz^{d}\mid |x-y|=1\}}\left(J_{xy}c_{i}^{\dagger}c_{j} + h_{xy}n_{x} n_{y}\right) \,,
\end{equation}
together with all its local truncations
\begin{equation*}
%\label{ffham1}
H_\L:=\sum_{\{(x,y)\in\L\mid |x-y|=1\}}\left(J_{xy}c_{i}^{\dagger}c_{j} + h_{xy}n_{x} n_{y}\right) \,,
\end{equation*}
where the $c_{x}$ and $c_{x}^{\dagger}$ are Fermion annihilators and
creators on the $x$--th site with the associated number operator $n_{x}:=c_{x}^{\dagger}c_{x}$. The coupling constants $J_{xy}$ and the external magnetic fields
$h_{xy}$ are independent
random variables, and we suppose that the $J_{xy}$, as well as the $h_{xy}$, are identically distributed on a symmetric bounded interval of the real line according to the laws $J(s)$, $h(s)$, respectively. Denote $E\bz^d$ the edges of the standard lattice $\bz^d$. Our sample space
$(\Om,\m)$ for the pivotal model described above has the form
\begin{equation}
\label{saspa}
\Om=\prod_{(x,y)\in E\bz^d}(\supp J\times\supp h)\,,\quad
\di\m=\prod_{(x,y)\in E\bz^d}(J(\di s)\times h(\di s))\,.
\end{equation}
First of all notice that the shift $\a_x$by $x\in\bz^d$ acts in a canonical way on the measurable space $(\Om,\m)$ just by shifting the edges in the trajectories,
 $$
 \om=\{(x_\om, y_\om)\}\mapsto T_x\om=\{(x_\om+x, y_\om+x)\}\,.
 $$
We have
\begin{prop}
\label{prop71}
Under the above notations, $\bz^d$ acts on $(\Om,\m)$ by a measure preserving mixing transformations.
\end{prop}
\begin{proof}
As the involved measure $\m$ is a product of a single measure $J(\di s)\times h(\di s)$ and $T_x$ is a bijection of $\Om$, it preserves $\m$. To check the ergodic properties of such an action, it is enough to reduce the matter to the measurable functions depending only by a finite number of variables. Let $f,g$ be two of such functions. If $|x|$ is sufficiently big, $f$ and $g$ depend on different sets of variables $\L_f$, $\L_g+x$ in the space made of the edges of $\bz^d$. Then we get
\begin{align*}
&\int f\,g\circ T_x\di\m=\int f\,g\circ T_x\di\m_{\L_f}\times\di\m_{\L_g+x}
=\int f\di\m_{\L_f}\int g\circ T_x\di\m_{\L_g+x}\\
&=\int f\di\m\int g\circ T_x\di\m
=\int f\di\m\int g\di\m\,.
\end{align*}
\end{proof}
Concerning the other useful regularity properties of the pivotal model described above, we need the following
\begin{lem}
\label{eehheh}
If $A\in\carf(\L)$ is localized in the bounded region $\L$, then
$f_{A,t}\in\carf(\bar\L)\otimes L^{\infty}(\Om,\m)$, where $\bar\L:=\{x\in\bz^d\mid\dist(x,\L)\leq1\}$.
\end{lem}
\begin{proof}
We have in this situation
$\t_{t}^{\om}(A)=e^{\imath H_{\bar\L}t}Ae^{-\imath H_{\bar\L}t}$. The proof follows by using the series expansion of the matrix $e^{\imath H_{\bar\L}t}$.
\end{proof}
Consider the map
$$
(t, \om)\in\br\times\Om\mapsto\t_{t}^{\om}\in\aut(\carf(\bz^d))\,,
$$
the last equipped with the $\s$--strong topology.\footnote{The two--sided uniform structure of the
$\s$--strong topology is generated by the countable family if semimetrics
\begin{equation}
\label{2un}
d_{\f}(\a,\b):=\|\f\circ\a-\f\circ\b\|+\|\f\circ\a^{-1}-\f\circ\b^{-1}\|
\end{equation}
where $\f$ runs on countable dense subset of $\cs(\ca)$.}
\begin{prop}
\label{prop73}
The one parameter group of random automorphism $\t_{t}^{\om}$ is jointly measurable in the variables $(t,\om)$.
\end{prop}
\begin{proof}
By taking into account the semimetrics \eqref{2un} which generate
the $\s$--strong topology, it is enough to check if all the
functions $(t, \om)\mapsto\f(\t_{t}^{\om}(A))$ are jointly
measurable, when $A$ and  $\f$ run over $\bigcup\carf(\L)$ and
$\bigcup\carf(\L)^*$ respectively, where $\L$ are all the bounded
subregions of $\bz^d$. As in Lemma \ref{eehheh} by expanding
$e^{\imath H_{\bar\L}t}$ in a power series, the functions mentioned
above can be expressed as series whose the terms are measurable
functions.
\end{proof}
\begin{prop}
\label{prop74}
For each $A\in\carf(\bz^d)$,
$f_{A,t}\in\carf(\bz^d)\otimes L^{\infty}(\Om,\m)$. In particular,
$f_{A,t}\in L^{\infty}(\Om,\m;\carf(\bz^d))$
\end{prop}
\begin{proof}
Let $\{A_n\}_{n\in\bn}\subset\bigcup_\L\carf(\L)$ be a sequence of localized elements converging to $A$. We get
$$
\|f_{A,t}-f_{A_n,t}\|_\infty\leq\|A-A_n\|\,.
$$
This means by Lemma \ref{eehheh}, that $f_{A,t}$ is uniform limit of measurable functions belonging to $\carf(\bz^d)\otimes L^{\infty}(\Om,\m)\subset L^{\infty}(\Om,\m;\carf(\bz^d))$. The proof follows as
$\carf(\bz^d)\otimes L^{\infty}(\Om,\m)$ is a closed subalgebra.
\end{proof}
The pivotal model considered above represents the disordered version
of the model considered for example in \cite{M2} (see also
\cite{M}), related to the investigation of the structure of the
ground states. Concerning  the temperature states, nothing is known
regarding the possible existence of the critical temperature(s),
even for the non disordered situation. Yet, it is possible to prove
some general properties concerning the structure of the the KMS
states. The reader is referred to \cite{AM1} for the non disordered
situation.\footnote{The reader is referred also to the papers
\cite{AFM, F1} for interesting connections which arise naturally
between the Markov property for Fermions and the KMS condition and
entanglement.}  By taking into account the propositions
\ref{prop71}, \ref{prop73}, \ref{prop74}, we can apply all the
results of the present paper to the disordered model based on
$\ga:=\carf(\bz^d)\otimes L^{\infty}(\Om,\m)$ and the Hamiltonian
\eqref{ffham}, where $(\Om,\m)$ is described by \eqref{saspa}. We
refer the reader to Section 5 of \cite{BF} for the proofs and
details.

It is a well known fact (i.e. a standard compactness trick, see e.g.
\cite{BR2}) that, for a fixed realization of the couplings
$\{J_{x,y}\}$, and the external magnetic field $\{h_{x,y}\}$, the
spin algebra $\carf(\bz^d)$ admits KMS states at each inverse
temperature $\b>0$. We start by considering in some detail the
uniqueness case. Such a situation arises if the quantum
model under consideration admits some critical temperature. The
situation is well clarified for many classical disordered models
(see e.g. \cite{Ne}), contrary to the quantum situation where, to
the knowledge of the authors, there are very few rigorous results
concerning this point, even for the standard model of quantum spin
glasses where the observables are modeled by the usual tensor
product of infinitely many copies of a full matrix algebra. Namely,
suppose that for a fixed $\b>0$, the Ising type model under
consideration admits a unique KMS state, say $\f_{\om}$, almost
surely. By the same arguments used in \cite{BF} we can show, thanks
to Proposition \ref{001ev}, that the map $\om\in\Om\mapsto
\f_{\om}\in\cs(\carf(\bz^d))$ is $*$--weak measurable and made of
even states almost surely. Furthermore, it satisfies almost surely
the condition of equivariance
\begin{equation}
\label{eqiv} 
\f_{\om}\circ\a_{x}=\f_{T_{-x}\om}
\end{equation}
w.r.t. the spatial translations, simultaneously. Namely, it defines by \eqref{a2} a state
$\f\in\cs_{NI}(\ga)$. Suppose now that $\psi$ is any KMS state at the inverse temperature $\b$, normal when restricted to $L^{\infty}(\Om,\m)$. Then, according to \eqref{a2}
$$
\psi=\int_\Om\psi_\om\di\m(\om)
$$
for a $*$--weak measurable field $\om\in\Om\mapsto \psi_{\om}\in\cs(\carf(\bz^d))$ of positive form. By Proposition \ref{001} and the uniqueness assumption, we get
$$
\psi_\om=\psi_\om(\idd)\f_\om\,,
$$
almost surely. We have then shown that there exists a one--to--one correspondence
$f\mapsto\f_f$ between positive normalized $L^{1}$--functions and KMS states for $\ga$ at inverse temperature $\b>0$. Namely,
\begin{equation}
\label{0001}
\f(A)=\int_{\Om}f(\om)\f_{\om}(A(\om))\di\m(\om)\,, \qquad
A\in\ga\,,
\end{equation}
where $f\in L^{1}(\Om,\m)$ is any positive normalized function. In a situation such as the one just described above, there is a
unique locally normal KMS state $\f$ on $\ga$ which is translation
invariant, which correspond to $f=1$ in \eqref{0001}, in addition, any KMS states is automatically even. Namely, $\f_f\in\cs_{N}(\ga)$. Finally,
there exists a unique $\l>0$ such that $\f_f$ is
a direct integral of $\ty{III_{\l}}$ factors almost surely.

We end the section by briefly describing what happens in ``multiple
phase'' regime, provided such a possibility exists for the model
under consideration. After taking the infinite--volume limit along
various subsequences $\L_{n_{k}}\uparrow\bz^{d}$, we will find, in
general, different locally normal translation invariant $\gpt$--KMS
states on $\ga$ at fixed inverse temperature $\b$, which are
automatically even. Fix one such a state $\f$. Then, one recovers a
$*$--weak measurable field
$\{\f_{\om}\}_{\om\in\Om}\subset\cs(\carf(\bz^d))$ of even
$\t^{\om}$--KMS states satisfying the equivariance property
\eqref{eqiv}. According to Proposition 3.1 of \cite{BF} (cf.
\cite{BR2}, Proposition 5.3.29), the set of the $\gpt$--KMS states
$\f_{T}\in\cs(\ga)$, locally normal w.r.t. $\f$, has the form
\begin{equation}
\label{000121}
\f_{T}(A)=\int_{\Om}
\big\langle\pi_{\f_{\om}}(A(\om))T(\om)^{1/2}\Psi_{\f_{\om}},
T(\om)^{1/2}\Psi_{\f_{\om}}\big\rangle_{\ch_{\f_{\om}}}\di\m(\om)\,.
\end{equation}
Here, $(\pi_{\f_{\om}},\ch_{\f_{\om}},\Psi_{\f_{\om}})$ is the GNS
representation of $\f_{\om}$, $\{T(\om)\}_{\om\in\Om}$ is a
measurable field of closed densely defined operators on
$\ch_{\f_{\om}}$ affiliated to the (isomorphic) centres
$\gz_{\f_{\om}}$ respectively, satisfying
$\Psi_{\f_{\om}}\in\cd_{T(\om)^{1/2}}$ almost surely, and
${\displaystyle
\int_{\Om}\|T(\om)^{1/2}\Psi_{\f_{\om}}\|_{\ch_{\f_{\om}}}^{2}\di\m(\om)=1}$. This means that
$\f_T$ is the direct integral of
$$
\f_{T(\om)}:=\big\langle\pi_{\f_{\om}}(A(\om))T(\om)^{1/2}\Psi_{\f_{\om}},
T(\om)^{1/2}\Psi_{\f_{\om}}\big\rangle_{\ch_{\f_{\om}}}\,.
$$
For physical application (cf. \cite{A2}), we specialize the situation when $\f_T$ is even. Again by
Proposition \ref{001ev}, this means that $\f_{T(\om)}$ is even, almost surely. It might be proven that it implies that $T$ is even, and then $T(\om)$ is even almost surely. In order to avoid technicalities due to the unboundedness of $T$ we prove the statement in the bounded case.
\begin{prop}
Suppose that the positive operator $T\eta\gz_{\f}$, describing the
even KMS state $\f_{T}$, which is normal w.r.t. $\f$, is bounded.
Then $T(\om)$ in \eqref{000121} is even, almost surely.
\end{prop}
\begin{proof}
Suppose $(\cb,\a)$ is a dynamical system, where $\cb$ is a
$\bz_2$--graded $C^*$--algebra and $\a$ is one parameter group of
even automorphisms of $\gb$. Let $\f$, $\psi$ be even $\a$--KMS
states of $\gb$ with $\psi$ normal w.r.t. $\f$. Consider the joint
covariant (relative to the time evolution $\a_t$, and the grading
$\s$) GNS representation $(\ch,\pi_\f,U_t,V,\F)$ of $\f$. By
\cite{BR2}, Proposition 5.3.29 there exists a unique positive
$T\eta\gz_{\f}$, with $\F\in\cd_{T^{1/2}}$ such that
$$
\psi(A)=\langle\pi_\f(A)T^{1/2}\F,T^{1/2}\F\rangle\,.
$$
Suppose now that $T$ is bounded.\footnote{The proof can be generalized to the unbounded situation, approximating $T$ by a sequence of bounded positive operators in $\gz_{\f}$. We leave the details to the reader.}. As $\f$ and $\psi$ are even, we get for each $A\in\gb$.
\begin{align*}
\langle\pi_\f(A)\F,T\F\rangle=&\psi(A)=\psi(\s(A))=\langle V\pi_\f(A)V\F,T\F\rangle\\
=&\langle \pi_\f(A)\F,VT\F\rangle=\langle \pi_\f(A)\F,VTV\F\rangle\,.
\end{align*}
By the cyclicity  of $\F$ we get $T\F=VTV\F$. As $\F$ is separating for $\gz_{\f}$ we conclude that
$T=VTV$, that is $T$ is even.
In our situation, the fact that $\f_T$ is supposed to be even and $T$ bounded, implies that the
$\f_{T(\om)}$
are even and the $T(\om)$ are bounded, almost surely.
The proof follows by applying the previous consideration, fiberwise to $\f_{T(\om)}$.
\end{proof}

\section*{Acknowledgements}
The second--named author (F.F.) would like to thank the Indian NBHM, Department of Atomic Energy, for the Visiting
Professorship, which provided him with an opportunity to begin work
on this problem while visiting Padre Conceicao  College of
Engineering, Verna, Goa. He acknowledges also the partial support of Italian INDAM--GNAMPA.
The author is grateful for the
warm hospitality during his stay at Padre Conceicao
College of Engineering, to 
the Indian Statistical Institute,
Bangalore, and the Institute of Mathematical Sciences, Chennai.

\end{document}